\documentclass[english]{article}
\usepackage[T1]{fontenc}
\usepackage[latin9]{inputenc}
\usepackage{float}
\usepackage{amsmath}
\usepackage{amsthm}
\usepackage{amssymb}
\usepackage{graphicx}
\usepackage{esint}

\makeatletter
\theoremstyle{plain}
\newtheorem{thm}{\protect\theoremname}
  \theoremstyle{definition}
  \newtheorem{defn}[thm]{\protect\definitionname}
  \theoremstyle{remark}
  \newtheorem{rem}[thm]{\protect\remarkname}
  \theoremstyle{plain}
  \newtheorem{prop}[thm]{\protect\propositionname}
  \theoremstyle{remark}
  \newtheorem*{rem*}{\protect\remarkname}

\makeatother

\usepackage{babel}
  \providecommand{\definitionname}{Definition}
  \providecommand{\propositionname}{Proposition}
  \providecommand{\remarkname}{Remark}
\providecommand{\theoremname}{Theorem}

\begin{document}

\title{A forward--backward random process for the spectrum of 1D Anderson
operators}

\author{Raphael Ducatez}
\maketitle
\begin{abstract}
We give a new expression for the law of the eigenvalues of the discrete
Anderson model on the finite interval $[0,N]$, in terms of two random
processes starting at both ends of the interval. Using this formula,
we deduce that the tail of the eigenvectors behaves approximately
like $\exp(\sigma B_{|n-k|}-\gamma\frac{|n-k|}{4})$ where $B_{s}$
is the Brownian motion and $k$ is uniformly chosen in $[0,N]$ independently
of $B_{s}$. A similar result has recently been shown by B. Rifkind
and B. Virag in the critical case, that is, when the random potential
is multiplied by a factor $\frac{1}{\sqrt{N}}$
\end{abstract}
We are interested in the one dimensional discrete Anderson model on
a finite domain $[0,N]$. This model is very classical and has been
studied extensively since the 70s. See for example the monograph of
Carmona Lacroix \cite{carmona2012spectral}. Compared to higher dimensions
case, it can be considered as a solved problem. However new approaches
can always shed new light on this famous system. 

The usual approach to tackle this system is the transfer matrix framework.
The eigenvectors of the random Schr\"odinger operator satisfy a recursive
relation of order 2, $u_{n+2}=(V_{n+1}-\lambda)u_{n+1}-u_{n}$, which
can be written in a matrix form. Using this relation, one can obtain
an eigenvector everywhere on $[0,N]$ from the product of the transfer
matrices applied to the boundary values. The advantage of such a formulation
is that one can then use the very powerful results for random matrices
product and from ergodic theory such as the Oseledets theorem. 

In the historical approach of Kunz and Souillard \cite{kunz1980spectre}
or in the proof from the book \cite{cycon2009schrodinger} a change
of variables is used to deal with the conditional probability of the
potential $V$ with a fixed eigenvalue $\lambda$. In this short note,
we propose another calculation of this conditional probability. We
define a random variable $k$ whose random law is close to the uniform
law on $[0,N]$. This variable splits the interval $[0,N]$ into two
part $[0,k]$ and $[k,n]$. On the left part, the matrices product
is made from left to right. On the right part, the matrices product
is made from right to left. And far from the cut, the laws of the
matrices are very close to be independent. 

The main interest of our approach is that the connection with the
theorems for products of random matrices is more transparent in this
setup. From this formula we can recover several known results. Relying
on the positivity of the Lyapunov exponent, the formula can be used as
a new proof of exponential Anderson localization of eigenvectors where
the center of localization is uniformly distributed on $[0,N]$. Moreover,
because it gives a explicit random law, we can go beyond the exponential
decay of the eigenvectors far from the center of localization and
give an explicit law for their tail. 

In the first section, we detail the model and we state our result.
Then we give some applications of our theorem in the second section.
In particular, we write an asymptotic result similar to the result
of Rifkind and Virag in \cite{rifkind2016eigenvectors}. In Section~3, 
we finally give the proof of the theorem.

\paragraph*{Acknowledgement: }

We would like to thank Mathieu Lewin for his encouragement, his interest
and his relevant comments. This project has received funding from
the European Research Council (ERC) under the European Union's Horizon
2020 research and innovation program (grant agreement MDFT No 725528).

\section{Model and main result}

We consider the discrete one dimensional Anderson model \cite{PhysRev.109.1492}
defined on $[1,N]$ through the operator
\[
H^{(N)}=-\Delta^{(N)}+V_{\omega}^{(N)}.
\]
Here $V_{\omega}^{(N)}$ is a random iid potential and 
\[
\Delta^{(N)}(x,y)=\begin{cases}
1 & \text{if \ensuremath{|x-y|=1}}\\
0 & \text{otherwise}
\end{cases}
\]
 is the usual discrete Laplacian. Hence $H$ is just the $N\times N$
symmetric matrix 
\[
H^{(N)}=\begin{pmatrix}V_{1} & 1\\
1 &  & \ddots\\
 & \ddots & \ddots & \ddots\\
 &  & \ddots &  & 1\\
 &  &  & 1 & V_{N}
\end{pmatrix}.
\]

We make the following assumption:

\begin{description}
\item[(H1)] The random law of $V_{\omega}$ is absolutely continuous with
respect to the Lebesgue measure. 
\end{description}

\subsection{Transfer matrices}

Transfer matrices have been one of the main tool to study the
1D Anderson model. One is interested in the eigenvectors, $(Hu)_{n}=\lambda u_{n}$,
which satisfy the recurrence relation 
\begin{equation}
\forall n\in[0,N],\quad-u_{n+1}+(V_{\omega}-\lambda)u_{n}-u_{n-1}=0,\label{eq:Eigenvector}
\end{equation}
with $u_{-1}=u_{N+1}=0$ such that the formula is valid for $n=0$
and $n=N$. This can be written with transfer matrices 
\[
\begin{pmatrix}u_{n+1}\\
u_{n}
\end{pmatrix}=T(v_{\omega}(n)-\lambda)\begin{pmatrix}u_{n}\\
u_{n-1}
\end{pmatrix}
\]
where
\[
\forall x,\quad T(x)=\begin{pmatrix}x & -1\\
1 & 0
\end{pmatrix}.
\]
We can then write the matrix product 
\[
M_{n}(\lambda)=\prod_{k=1}^{n}T(v_{\omega}(k)-\lambda)
\]
and we have 
\[
\begin{pmatrix}u_{n+1}\\
u_{n}
\end{pmatrix}=M_{n}(\lambda)\begin{pmatrix}1\\
0
\end{pmatrix}.
\]
The parameter $\lambda$ is an eigenvalue if and only if there exist
$c\in\mathbb{R}$ such that 
\[
M_{N}(\lambda)\begin{pmatrix}1\\
0
\end{pmatrix}=\begin{pmatrix}0\\
c
\end{pmatrix},
\]
the condition $u_{N+1}=0$ is then satisfied. 

It will be convenient to denote the vector $\begin{pmatrix}u_{n+1}\\
u_{n}
\end{pmatrix}$ as a complex number in the fashion 
\[
u_{n+1}+iu_{n}=z_{n}=r_{n}e^{i\phi_{n}}
\]
where $r_{n}\in\mathbb{R}_{+}\text{ and }\phi_{n}\in\mathbb{R}/2\pi\mathbb{Z}$.
We also introduce the lifting of $\phi_{k}$, which we denote by $\theta_{k}$.
This is just a discrete version of the continuous lifting from $\mathbb{R}/2\pi\mathbb{Z}$
to $\mathbb{R}$ into the discrete case. It is defined recursively
by 
\[
\text{ }\theta_{k}=\begin{cases}
\theta_{0}=0\\
\phi_{k}\ [2\pi] & \forall k\in[0,N]
\end{cases}
\]
and

\[
\theta_{k}-\frac{\pi}{2}\leq\theta_{k+1}<\theta_{k}+\frac{3\pi}{2}.
\]
It can be seen that $\phi_{k+1}$ does not depend on $r_{k+1}$ but
only on $\phi_{k}$ and $T_{\lambda}(v_{k}-\lambda)$. Therefore,
for simplicity of notation, we use the same notation $T$ for the
operator on $\mathbb{R}/2\pi\mathbb{Z}$: $\phi_{k+1}=T_{\lambda}(v_{k})\phi_{k}$
.

Note that it is possible to recover $r_{k}$ from $\phi_{0},\phi_{1},...,\phi_{N}$
with the formula 
\[
\frac{r_{k+1}}{r_{k}}=\frac{r_{k+1}}{u_{k+1}}\frac{u_{k+1}}{r_{k}}=\frac{\cos\phi_{k}}{\sin\phi_{k+1}}.
\]
 For this reason, in the rest of the paper we focus mostly on $(\phi_{k})_{k=0\dots N}$.
We note $\mathcal{F}(\lambda)=(\phi_{k})_{k=0\dots N}$ which has
been constructed from the recursive formula $\phi_{k+1}=T_{\lambda}(v_{k})\phi_{k}$
and $\phi_{0}=0$. And for $\lambda$ an eigenvalue, we note $\mathcal{P}h(\lambda)=(\phi_{k})_{k=0\dots N}$
the phase of the corresponding eigenvector. Note that it is equal
to $\mathcal{F}(\lambda)$ with the condition $\phi_{N}=\frac{\pi}{2}\;[\pi]$.

\subsection{Forward and backward processes}

In this subsection, we define two natural random laws on the chain
$X=(\phi_{k})_{k=0,..,N}$. The first one is the Markov chain starting
from $\phi_{0}$ with an initial law $\mu_{f}$ defined on $\mathbb{S}^{1}$
and transition law $\phi_{k}\rightarrow\phi_{k+1}=T(v_{k})\phi_{k}$
with a random measure $\nu$ for $v_{k}$. We call it the forward
process. The second one is the Markov chain starting from $\phi_{N}$
with an initial law $\mu_{b}$ and transition law $\phi_{k}\rightarrow\phi_{k-1}=T_{\lambda}^{-1}(v_{k-1})\phi_{k}$
with a random measure $\nu$ for $v_{k}$ and we call it the backward
process. Then we introduce a cut in $[0,N]$, and we can define the
random law product between these two processes which we call the forward--backward
process. 

For a proper definition we use test functions on $\mathbb{R}^{N+1}$
which are bounded and continuous. 

\begin{defn}[Forward and backward processes]
The probability $\mathcal{P}_{f}$ on
$\mathbb{R}^{N+1}$, defined by 
\[
\mathcal{P}_{f}(F)=\idotsint d\mu_{f}(\phi_{0})d\nu(v_{1})\cdots d\nu(v_{n})F(X)
\]
for any test function $F$, is called the \emph{forward process}.
Similarly,  the probability $\mathcal{P}_{b}$ on $\mathbb{R}^{N+1}$ defined by 
\[
\mathcal{P}_{b}(F)=\idotsint d\nu(v_{1})...d\nu(v_{N})d\mu_{b}(\phi_{N})F(X)
\]
for any test function $F$, is called the \emph{backward process}.
\end{defn}

\begin{rem}
\label{rem:Change}If we introduce $\xi_{0,n}:\phi_{0}\rightarrow\phi_{n}^{f}=\prod_{k=0}^{n-1}T(v_{k})\phi_{0}$
and if for almost surely any $v_{1},v_{2},...,v_{n}$, $\mu_{b}$
and the push measure $\xi(\mu_{f})$ are equivalent measures, then
we remark that for any $F$: 

$\begin{aligned}\mathcal{P}_{b}(F) & =\idotsint d\nu(v_{1})...d\nu(v_{n})d\mu_{f}(X_{0})\frac{d\mu_{b}(X_{n})}{d\xi(\mu_{f}(X_{0}))}\big|_{v_{1},...,v_{n}}F(X)\\
 & =\mathcal{P}_{f}\Bigg(F\frac{d\mu_{b}(X_{n})}{d\xi(\mu_{f}(X_{0}))}\big|_{v_{1},...,v_{n}}\Bigg).
\end{aligned}
$\end{rem}

\begin{defn}[Forward-Backward process]
For $k\in[0,N]$, we define $\mathcal{P}_{f,0..k}\otimes\mathcal{P}_{b,k+1,...,N}$
a forward process for $X^{f}=\phi_{0}^{f},\phi_{1}^{f}...,\phi_{k}^{f}$
with $\phi_{0}^{f}=0$, ($\mu_{f}=\delta_{0}$) and a backward process
for $X^{b}=\phi_{N}^{b},\phi_{N-1}^{b}...,\phi_{k}^{b}$, with $\phi_{N}^{b}=\frac{\pi}{2}$
($\mu_{b}=\delta_{\frac{\pi}{2}}$) which are independent from each
other.
\end{defn}

\subsection{Main result}

We are now ready to state the main theorem of our paper.
\begin{thm}[Law of the spectrum of the 1D Anderson model]
For any test function
$G(\lambda,X)$, we have \label{thm:(Change-of-direction)}
\begin{align}
 & \mathbb{E}\Big[\sum_{\lambda\in\sigma(H),X=\mathcal{P}h(\lambda)}G(\lambda,X)\Big]\nonumber \\
 & \qquad=\int_{\mathbb{R}}d\lambda\sum_{k=1}^{N}\mathbb{E}_{\mathcal{P}_{f,1..k}\otimes\mathcal{P}_{b,k+1,...,N}}\Big[G(\lambda,X)\delta_{\phi_{k}^{f}-\phi_{k}^{b}[\pi]}\sin^{2}(\phi_{k}^{f})\Big]\label{eq:Change-of-direction}
\end{align}
that we can rewrite as
\begin{align}
 & \mathbb{E}\Big[\frac{1}{N}\sum_{\lambda\in\sigma(H),X=\mathcal{P}h(\lambda)}G(\lambda,X)\Big]\nonumber \\
 & \qquad=\int_{\mathbb{R}}\rho(\lambda)d\lambda\Bigg(\frac{1}{N}\sum_{k=1}^{N}\mathbb{E}_{\mathcal{P}_{f,1..k}\otimes\mathcal{P}_{b,k+1,...,N}}\Big[G(\lambda,X)\frac{\delta_{\phi_{k}^{f}-\phi_{k}^{b}[\pi]}\sin^{2}(\phi_{k}^{f})}{\rho(\lambda)}\Big]\Bigg)\label{eq:Change-of-direction-1}
\end{align}
with $\rho(\lambda)$ the density of state.
\end{thm}

Recall that $Ph(\lambda)$ is the phase of the eigenvector corresponding
to the eigenvalue $\lambda$.

This formula is to be understood as follows. One chooses $k$ randomly
in $[1,N]$ which splits the segment into two parts $[1,k]$ and $[k,N]$.
On the left, we obtain a forward process, on the right, we obtain
a backward process. The choice of $k$ is not exactly uniform on $[0,N]$
because of the condition $\delta_{\phi_{k}^{f}-\phi_{k}^{b}}\sin^{2}(\phi_{k})$.
However, for large $N$, and for any $k\leq N$ not too close to $0$
or $N$, the laws of $\phi_{k}^{f}$ and $\phi_{k}^{b}$ are very close
to their invariant measure and then do not depend on $k$. Therefore
the law of $k$ becomes close to the uniform.

There is still a dependence between the two processes at the connection
between the forward and backward processes. However, because of the mixing
property of the matrix product, the correlations decay exponentially
fast away from the cut $k$. 

We recall that a stationary process $X_{k}$ is called $(\alpha_{n})_{n\in\mathbb{N}}-$
mixing if
\[
\forall k,\quad\max_{A,B}|\mathbb{P}(X_{k}\in A,X_{k+n}\in B)-\mathbb{P}(X_{k}\in A)\mathbb{P}(X_{k+n}\in B)|\leq\alpha_{n}
\]

The following is a well known result.
\begin{prop}
There exists a constant $C>0\text{ and }0<\kappa<1$ such that the
process $\phi_{k}$ is $(C\kappa^{n})_{n\in\mathbb{N}}-$mixing.\label{lem:mixing} 
\end{prop}
For a proof, see \cite[proposition IV.3.12]{carmona2012spectral}.

\section{Applications}

We present here three application of our result. The first one is
a formula for the integrated density of states. The second one is
about the form of the tails of the eigenvectors. We then finish with
a temperature profile from \cite{2017JSP...tmp...71D}.

\subsection{A formula for the integrated density of states}

The following equality can be found as well in \cite{carmona2012spectral} (proposition
VIII.3.10 and problem VIII.6.8). 
\begin{prop}
For $\lambda\in\mathbb{R}$, let $\mu_{\lambda}(d\phi)=m_{\lambda}(\phi)d\phi$
be the $T_{\lambda}$-invariant measure on $\mathbb{R}/\mathbb{Z}$.
The density of states 
\[
dN(\lambda)=\lim_{N\rightarrow\infty}\frac{1}{N}\#\big\{\sigma(H)\cap[\lambda,\lambda+d\lambda]\big\}
\]
 is given by 
\[
\frac{dN(\lambda)}{d\lambda}=\int_{\mathbb{\mathbb{R}}/2\mathbb{\pi\mathbb{Z}}}\sin^{2}(\phi)m_{\lambda}(\phi)m_{\lambda}\Big(\frac{\pi}{2}-\phi\Big)d\phi.
\]
\end{prop}
\begin{proof}
We apply our formula~\eqref{eq:Change-of-direction-1} in Theorem \ref{thm:(Change-of-direction)}. We choose
$G(s,X)=G(s)$ (that does not depend on $X$) and recognize $\int_\mathbb{R} G(\lambda)\rho(\lambda)d\lambda$. More precisely,
\begin{align*}
\frac{1}{N}\mathbb{E}\Big[\sum_{\lambda\in\sigma(H_{N})}G(\lambda)\Big] & =\int G(\lambda)d\lambda\frac{1}{N}\sum_{k}\mathbb{E}_{\mathcal{P}_{f,1..k}\otimes\mathcal{P}_{b,k+1,...,N}}[\delta_{\phi_{k}^{f}-\phi_{k}^{b}}\sin^{2}(\phi_{k})]\\
 & =\int G(\lambda)d\lambda\frac{1}{N}\sum_{k}\int_{\phi}\rho_{k,\lambda}(\phi)\rho_{N-k,\lambda}(\frac{\pi}{2}-\phi)\sin^{2}(\phi_{k})]
\end{align*}
where $\rho_{k,\lambda}$ and $\rho_{N-k,\lambda}$ are the density
probabilities of the angles of $M_{k}(\lambda)\begin{pmatrix}1\\
0
\end{pmatrix}$ and $M_{N-k}(\lambda)\begin{pmatrix}1\\
0
\end{pmatrix}$. We can then conclude using that $\rho_{k,\lambda}\rightarrow m_{\lambda}$
and $\rho_{N-k,\lambda}\rightarrow m_{\lambda}$ when $k\rightarrow\infty$
and $N-k\rightarrow\infty$.
\end{proof}

\subsection{Brownian and drift for the eigenvectors }

It is well known since the work of Carmona-Klein-Martinelli~\cite{carmona1987anderson},
Goldsheild-Molchanov-Pastur \cite{goldsheid1977random} and Kunz-Souillard
\cite{kunz1980spectre} that the eigenvectors are localized and decay
exponentially from the center of localization. An exact form of the
eigenvectors has been recently proven in the critical case where $V$
is replaced by $\frac{V}{\sqrt{N}}$ in \cite{rifkind2016eigenvectors}.
The authors proved that the eigenvectors in the bulk have the form
$e^{\sigma\frac{B_{|t-u|}}{2}-\gamma|t-u|}$. We claim using our formula
of Theorem \ref{thm:(Change-of-direction)} that a similar result
holds for the tails of the eigenvectors in the non critical
case. 

For the reader's convenience we recall the heuristics of the following
classical results. One can write any product of random matrices $M_{N}=\prod_{i=1}^{N}T_{i}$
as 
\[
\log(\|M_{N}\|)=\log\Big(\prod_{i=1}^{N}\frac{\|M_{i}\|}{\|M_{i-1}\|}\Big)=\sum_{i=1}^{N}\log\Big(\|T_{i}\big(\frac{M_{i-1}}{\|M_{i-1}\|}\big)\|\Big)
\]
In the case when $T_{i}$ are iid and there are some strong mixing
property on $\frac{M_{i-1}}{\|M_{i-1}\|}$, the terms $Y_{i}=\log\Big(\|T_{i}\big(\frac{M_{i-1}}{\|M_{i-1}\|}\big)\|\Big)$
should behave like iid random variables. One can then prove the strong
law of large numbers, the central limit theorem, and Donsker's theorem.
See the paper of Le Page \cite{le1982theoremes} for these results.
One therefore defines a ``mean'', a ``variance'' and a ``random
walk`` as follows.

\begin{defn}
The \emph{Lyapunov exponent} is 
\[
\gamma(\lambda):=\lim_{N\rightarrow\infty}\frac{1}{N}\mathbb{E}\big[\log\|M_{N}(\lambda)\|\big].
\]
The \emph{limit variance} is
\[
\sigma^{2}(\lambda)=\lim_{N\rightarrow\infty}\frac{1}{N}\mathbb{E}\big[(\log\|M_{N}(\lambda)\|-\gamma(\lambda)N)^{2}\big].
\]
The \emph{random walk} is
\[
S_{n}=\frac{1}{\sigma(\lambda)}\big(\log\|M_{n}(\lambda)\|-\gamma(\lambda)n\big).
\]
and we consider its rescaling
\[
W_{N}(t)=\frac{1}{\sqrt{N}}S_{\lfloor Nt\rfloor}.
\]
\end{defn}

Finally, we denote by $W$ the Wiener measure.
\begin{thm}[Limit theorem for products of random matrices]\label{thm:LEPAGE}
We have the following:
\begin{itemize}
\item $\gamma(\lambda)>0$ and $\lim_{N\rightarrow\infty}\frac{1}{N}\log\|M_{N}(\lambda)\|=\gamma$,
almost surely;
\item $\sigma^{2}(\lambda)>0$\label{thm:Sigma};
\item $W_{N}\rightarrow W$ in law.
\end{itemize}
\end{thm}
We refer to \cite[Theorems 2 and 3]{le1982theoremes} for the proof
of Theorem \ref{thm:LEPAGE}.

We recover then the form of Brownian with drift and both on the right
hand side and the left hand side of the cut. For $\lambda$ an eigenvalue,
and $r_{k}e^{i\phi_{k}}$ the corresponding eigenvector, we note $q_{k}^{\lambda}=\log(r_{k})$.
For scaling, we set $q^{\lambda}(s)=q_{\lfloor Ns\rfloor}^{\lambda}/N$.

\begin{prop}[Tail of eigenvectors]\label{prop:(Tail-of-eigenvectors)}
1) Choosing $\lambda^{(N)}$ uniformly in $\sigma(H^{(N)})$, we have
the following convergence in law 

\[
(\lambda^{(N)},q_{s}^{\lambda^{(N)}})\rightarrow_{N\infty}(\tilde{\lambda},-|\gamma(\tilde{\lambda})(s-x)|)
\]
where $\tilde{\lambda}$ is a random variable with law the limiting
density of state $\rho$ and $x$ an independent variable on $[0,1]$
with uniform law.

\smallskip

2) There exists a sequence of random variables $\{x^{(N)}\}$
with uniform law on $[0,1]$ such that

\[
(\lambda^{(N)},\sqrt{N}\big[q_{s}^{\lambda^{(N)}}+|\gamma(\lambda^{N})(s-x^{(N)})|\big])\rightarrow_{N\infty}(\tilde{\lambda},\sigma(\tilde{\lambda})W_{s-x})
\]
where $W$ is the Wiener measure.
\end{prop}

The first statement is the very classical result of Anderson localization
for the one dimensional model. The eigenvectors decay exponentially
from their center of localization and this center is chosen uniformly
on the domain. The second statement says that the typical deviation
from the decay is the exponential of a Brownian (see Figure \ref{FigureDirichlet}
for an illustration).

Rifkind and Virag \cite{rifkind2016eigenvectors} studied the large
eigenvectors of the one dimensional Anderson model in the continuous
case where the potential is a white noise. It is the limit of the
discrete model in the ``critical regime'' where the potential is
scaled like $V_{\omega}^{(N)}=\frac{1}{\sqrt{N}}V_{\omega}$. In this
regime, one cannot speak of localization because the length
of the decay is as large as the size of the domain. However they proved
the exact law of the form of the eigenvectors 
\[
q_{s}^{\lambda}=-|\gamma(\lambda)(s-x)|+\sigma(\lambda)W_{s-x}.
\]

To make the connection with our previous proposition, one can actually
show that for $V_{\omega}=\epsilon v_{\omega}$, with $\mathbb{E}(v_{\omega}^{2})=\sigma^{2}$,
in the limit $\epsilon\rightarrow0$ and $|\lambda|<2,$ we have

\[
\gamma(\lambda)=\frac{\sigma^{2}}{4-\lambda^{2}}\epsilon^{2}
\]
 and 
\[
\frac{\sigma(\lambda)^{2}}{2}=\frac{\sigma^{2}}{4-\lambda^{2}}\epsilon^{2}.
\]

\begin{figure}
\centering{}\includegraphics[scale=0.4]{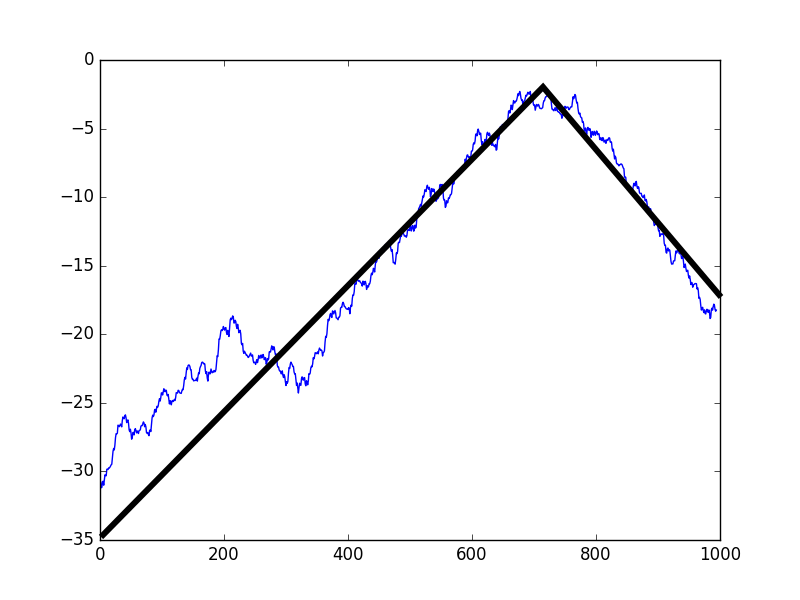}

\caption{A realization of $\log\|M_{n}(\lambda)\|$ for $N=1000$, $v_{\omega}$
uniform on $[0,1]$ with Dirichlet boundary conditions. we add a fit
of the form $|\gamma(\lambda)(s-x)|$.\label{FigureDirichlet}}
\end{figure}

\begin{proof}[Proof of Proposition \ref{prop:(Tail-of-eigenvectors)}]

If in our formula (\ref{eq:Change-of-direction-1}) the term $\delta_{\phi_{k}^{f}-\phi_{k}^{b}}$
were not there, then the forward and the backward processes would be completely independent. Our proposition would have then immediately
followed from Theorem \ref{thm:LEPAGE}, under the conditions that
$r_{k}$ obtained by the forward process and the $r_{k}$ obtained
by the backward process are the same, and that the normalization $\sum_{n=1}^{N}|u_{n}|^{2}=1$ holds. The latter becomes in the limit $\sup q_{s}^{\lambda}=0$.

Therefore we only have to check that the little perturbation around
the cut $k$ has no influence. We fix $\phi$. Conditionally of $\phi_{k}^{b}=\phi$
and $\phi_{k}^{f}=\phi$ the forward and backward processes are independent. The results
of Theorem \ref{thm:LEPAGE} are true asymptotically with probability
$1$. Therefore for any $\phi$ in a set of full Lebesgue measure
in $\mathbb{S}^{1}$ the results of Theorem \ref{thm:LEPAGE} are
true conditionally of $\phi_{k}^{b}=\phi$ and $\phi_{k}^{f}=\phi$.
\end{proof}

\subsection{A temperature profile}

We will use our result to explain some numerical observations which
have been made in \cite{2017JSP...tmp...71D}. In this article, the
authors are interested in the temperature profile of a disordered
chain connected to two thermal baths of temperatures $T_{0}$ and $T_{N}$
at the boundary $0$ and $N$. According to \cite{2017JSP...tmp...71D},
the temperature $T(x)$ at site $x$ is expected to be given in a
certain limit by 

\begin{equation}
T(x)=\sum_{\lambda\in\sigma(H)}|\psi_{\lambda}(x)|^{2}\Big(T_{0}\frac{|\psi_{\lambda}(0)|^{2}}{|\psi_{\lambda}(0)|^{2}+|\psi_{\lambda}(N)|^{2}}+T_{N}\frac{|\psi_{\lambda}(N)|^{2}}{|\psi_{\lambda}(0)|^{2}+|\psi_{\lambda}(N)|^{2}}\Big)\label{eq:TemperatureProfile}
\end{equation}
 where $H$ is our one--dimensional random Schrödinger operator and
$\psi_{\lambda}$ are its eigenvectors.

We prove that $T$ converge to a step function where the transition
from $T_{0}$ and $T_{N}$ happens in a neibourghood of $x=\frac{N}{2}$
at a scale $\sqrt{N}$. This has been observed numerically in \cite{2017JSP...tmp...71D}.

\begin{prop}
For $N$ large enough we have
\[
\lim_{N\rightarrow\infty}\mathbb{E}\big[T(\lfloor\sqrt{N}x+\frac{N}{2}\rfloor)\big]=T_{0}+(T_{N}-T_{0})\int_{\mathbb{R}}\mathbb{P}\Big(\mathcal{N}(0,1)\leq\frac{2\gamma(\lambda)}{\sigma(\lambda)}x\Big)dN(\lambda)
\]
 where $dN(\lambda)$ is the integrated density of states, $\gamma(\lambda)$ is
the Lyapunov exponent and $\sigma(\lambda)$ is the limit variance.
\end{prop}

The Lyapunov exponent is positive, continuous and so is bounded below on the support of $\sigma(H)$. The variance $\sigma(\lambda)$
is bounded as well. Therefore, uniformly in $\lambda$, $\mathbb{P}(\mathcal{N}(0,1)\leq\frac{\sigma(\lambda)}{2\gamma(\lambda)}x)\rightarrow0$
for $x\rightarrow-\infty$ and $\mathbb{P}(\mathcal{N}(0,1)\leq\frac{\sigma(\lambda)}{2\gamma(\lambda)}x)\rightarrow1$
for $x\rightarrow\infty$. We have then $T(y)=T_{0}$ for $\frac{N}{2}-y\gg\sqrt{N}$
and $T(y)=T_{N}$ for $y-\frac{N}{2}\gg\sqrt{N}$. This is the step function numerically observed in~\cite{2017JSP...tmp...71D}.

\begin{proof}
We use our formula and write: 
\begin{multline*}
\mathbb{E}(T(x)) =\sum_{k\in[0,N]}\int_{\mathbb{R}}d\lambda\mathbb{E}_{\mathcal{P}_{f,1..k}\otimes\mathcal{P}_{b,k+1,...,N}}\Big[|\psi_{\lambda}(x)|^{2}\\
\left(T_{0}\frac{|\psi_{\lambda}(0)|^{2}}{|\psi_{\lambda}(0)|^{2}+|\psi_{\lambda}(N)|^{2}}+T_{N}\frac{|\psi_{\lambda}(N)|^{2}}{|\psi_{\lambda}(0)|^{2}+|\psi_{\lambda}(N)|^{2}}\right)\delta_{\phi_{k}^{f}-\phi_{k}^{b}}\sin^{2}(\phi_{k})\Big] 
\end{multline*}
With the notation of Proposition \ref{prop:(Tail-of-eigenvectors)},
we write
\begin{multline*}
  T_{0}\frac{|\psi_{\lambda}(0)|^{2}}{|\psi_{\lambda}(0)|^{2}+|\psi_{\lambda}(N)|^{2}}+T_{N}\frac{|\psi_{\lambda}(N)|^{2}}{|\psi_{\lambda}(0)|^{2}+|\psi_{\lambda}(N)|^{2}}\\
 =T_{0}\frac{e^{Nq_{0}^{\lambda^{(N)}}}}{e^{Nq_{0}^{\lambda^{(N)}}}+e^{Nq_{1}^{\lambda^{(N)}}}}+T_{1}\frac{e^{Nq_{1}^{\lambda^{(N)}}}}{e^{Nq_{0}^{\lambda^{(N)}}}+e^{Nq_{1}^{\lambda^{(N)}}}}.
\end{multline*}
 Therefore in the limit $N\to\infty$, this converges to $T_{0}$ for $q_{0}^{\lambda}>q_{1}^{\lambda}$
and $T_{1}$ for $q_{0}^{\lambda}<q_{1}^{\lambda}$. We have then
at the limit a Bernoulli $T_{int}$ with parameter given by Proposition
\ref{prop:(Tail-of-eigenvectors)}:
\[
T_{int}=\begin{cases}
T_{0} & \text{with probability }\mathbb{P}\left(\mathcal{N}(0,1)\leq\frac{(2k-N)\gamma(\lambda)}{\sqrt{N}\sigma(\lambda)}\right),\\
T_{N} & \text{with probability }\mathbb{P}\left(\mathcal{N}(0,1)\geq\frac{(2k-N)\gamma(\lambda)}{\sqrt{N}\sigma(\lambda)}\right).
\end{cases}
\]
 In order to conclude, we recall that the whole mass of $|\psi_{\lambda}|^{2}$
is around a few number of sites around $k$ so 
\begin{multline*}
\mathbb{E}(T(x)) =\int_{\mathbb{R}}d\lambda\sum_{k\in[x-\alpha(N),x+\alpha(N)]}\mathbb{E}_{\mathcal{P}_{f,1..k}\otimes\mathcal{P}_{b,k+1,...,N}}\Bigg[|\psi_{\lambda}(x)|^{2}\\
\left(T_{0}\frac{|\psi_{\lambda}(0)|^{2}}{|\psi_{\lambda}(0)|^{2}+|\psi_{\lambda}(N)|^{2}}+T_{N}\frac{|\psi_{\lambda}(N)|^{2}}{|\psi_{\lambda}(0)|^{2}+|\psi_{\lambda}(N)|^{2}}\right)\delta_{\phi_{k}^{f}-\phi_{k}^{b}}\sin^{2}(\phi_{k})\Bigg]\\
 +O\left(e^{-\gamma(\lambda)\alpha(N)}\right),
\end{multline*}
where we have chosen $\alpha(N)$ such that $\sqrt{N}\gg\alpha(N)\text{\ensuremath{\gg}1}$.
Moreover for large $N$, 
\[
\mathbb{P}(\mathcal{N}(0,1)\geq\frac{(2x-N)\gamma(\lambda)}{\sqrt{N}\sigma(\lambda)})=\mathbb{P}(\mathcal{N}(0,1)\geq\frac{(2k-N)\gamma(\lambda)}{\sqrt{N}\sigma(\lambda)})+o(1),
\]
and we have then
\begin{multline*}
\mathbb{E}(T(x)) =\int_{\mathbb{R}}d\lambda\sum_{k\in[x-\alpha(N),x+\alpha(N)]}\mathbb{E}_{\mathcal{P}_{f,1..k}\otimes\mathcal{P}_{b,k+1,...,N}}\Bigg[|\psi_{\lambda}(x)|^{2}\\
\left(T_{0}+(T_{N}-T_{0})\mathbb{P}\left(\mathcal{N}(0,1)\geq\frac{(2x-N)\gamma(\lambda)}{\sqrt{N}\sigma(\lambda)}\right)\right)\delta_{\phi_{k}^{f}-\phi_{k}^{b}}\sin^{2}(\phi_{k})\Bigg]+o(1) 
\end{multline*}
Finally we use the following formula, for $x$ not close to the edges

\[
\sum_{k\in[0,N]}\mathbb{E}_{\mathcal{P}_{f,1..k}\otimes\mathcal{P}_{b,k+1,...,N}}[|\psi_{\lambda}(x)|^{2}\delta_{\phi_{k}^{f}-\phi_{k}^{b}}\sin^{2}(\phi_{k})]=\frac{dN(\lambda)}{d\lambda}+o(1).
\]
Indeed, for any $A$ Borel set of $\mathbb{R}$,
\begin{align*}
 & \int_{\mathbb{R}}1_{A}(\lambda)\frac{dN(\lambda)}{d\lambda}d\lambda\\
 & \qquad=\lim\frac{1}{N}\mathbb{E}(Tr(1_{A}(H)))\\
 & \qquad=\frac{1}{N}\sum_{x}\mathbb{E}\Big[\sum_{\lambda\in A\cap\sigma(H)}|\psi_{\lambda}(x)|^{2}\Big]\\
 & \qquad=\int1_{A}(\lambda)d\lambda\frac{1}{N}\sum_{x}\sum_{k\in[0,N]}\mathbb{E}_{\mathcal{P}_{f,1..k}\otimes\mathcal{P}_{b,k+1,...,N}}\Big[|\psi_{\lambda}(x)|^{2}\delta_{\phi_{k}^{f}-\phi_{k}^{b}}\sin^{2}(\phi_{k})\Big].
\end{align*}
We then note that the left term is asymptotically independent of $x$
for $x$ not close to the edges. Therefore 
\begin{align*}
 & \int_{\mathbb{R}}1_{A}(\lambda)\frac{dN(\lambda)}{d\lambda}d\lambda\\
 & \qquad=\int1_{A}(\lambda)d\lambda\sum_{k\in[0,N]}\mathbb{E}_{\mathcal{P}_{f,1..k}\otimes\mathcal{P}_{b,k+1,...,N}}[|\psi_{\lambda}(x)|^{2}\delta_{\phi_{k}^{f}-\phi_{k}^{b}}\sin^{2}(\phi_{k})]+o(1).
\end{align*}
The proposition then follows, namely we have
\[
\mathbb{E}(T(x))=T_{0}+\int_{\mathbb{R}}(T_{N}-T_{0})\mathbb{P}\Big(\mathcal{N}(0,1)\geq\frac{(2x-N)\gamma(\lambda)}{\sqrt{N}\sigma(\lambda)}\Big)\frac{dN(\lambda)}{d\lambda}d\lambda+o(1)
\]
as we wanted.
\end{proof}

\subsection{Periodic boundary conditions}

We have tried to obtain a similar result for periodic boundary conditions.
With the multiscale analysis tools \cite{frohlich1983}, one has the
exponential decay from the center of localization. But it would be
also interesting to have an interpretation with forward backward process
in this case.

In the critical regime, one would expect the form of the eigenvectors
to be like $e^{F(s)}$, on $\mathbb{R}/2\pi\mathbb{Z}$ with $F(s)=-\gamma\min(|s-u|,|s-u+\pi|)+\sigma\tilde{B}_{s-u}$
with $u$ uniformly chosen on $[0,2\pi]$ and $\tilde{B}$ a Brownian
bridge. So far we have not been able to prove this statement rigorously., but our intuition seems to be confirmed by numerical simulations (see Figure~\ref{figurePeriodic}).

\begin{rem*}
The condition $u_{-1}=u_{N+1}=0$ in the Dirichlet case has to be
replaced by $\text{Tr}(M_{N}(\lambda))=2$. Indeed, let $u_{n}$ be
an eigenvector of eigenvalue $\lambda$ and $z=\begin{pmatrix}u_{1}\\
u_{0}
\end{pmatrix}$. Then, periodic boundary conditions mean $M_{N}(\lambda)z=z$. So
$1$ is an eigenvalue of $M_{N}(\lambda)$. Therefore $1$ is a solution
of $x^{2}-Tr(M_{N}(\lambda))x+1=0$ and so $\text{Tr}(M_{N}(\lambda))=2$.
\end{rem*}

\begin{figure}[H]
\centering{}\includegraphics[scale=0.4]{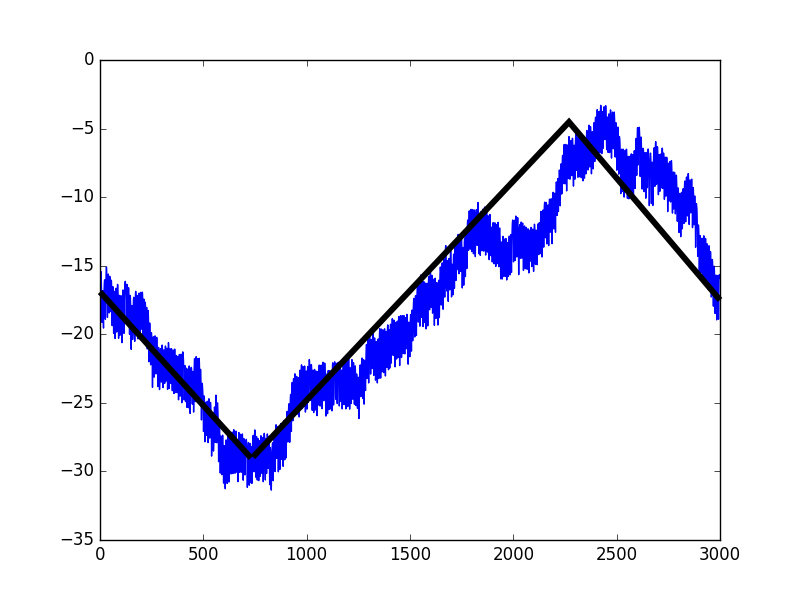}

\caption{A realization of $\log\|M_{n}(\lambda)\|$ with periodic boundary
conditions for $N=3000$, $v_{\omega}$ uniform on $[0,0.3]$. We
add a fit of the form $-\gamma\min(|s-u|,|s-u+\pi|)$.\label{figurePeriodic}}
\end{figure}

\section{Proof of Theorem \ref{thm:(Change-of-direction)}}
\begin{proof}
We have that $\lambda$ is an eigenvalue if and only if $\phi_{N}=\frac{\pi}{2}[\pi]$,
therefore
\[
\mathbb{E}\Big[\sum_{\lambda\in\sigma(H),X=\mathcal{P}h(\lambda)}G(\lambda,X)\Big]=\mathbb{E}\Big[\sum_{\lambda:\theta_{N}(\lambda)\in\frac{\pi}{2}+\pi\mathbb{Z},\text{ }X=\mathcal{F}(\lambda)}G(\lambda,X)\Big].
\]

\begin{rem}
$\theta_{N}:\lambda\rightarrow\theta_{N}(\lambda)$ is continuous
and strictly increasing (see the calculations below). 
\end{rem}
For finite $N$, the inverse function $\theta_{N}^{-1}$ is continuous,
so are $G$, $X$. We can therefore write
\begin{multline*}
\mathbb{E}\Big[\sum_{\lambda:\theta_{N}(\lambda)\in2\pi\mathbb{Z},X=\mathcal{F}(\lambda)}G(\lambda,X)\Big]\\
=\lim_{\epsilon\rightarrow0}\mathbb{E}\left[\frac{1}{2\epsilon}\sum_{n\in\mathbb{Z}}\int_{2\pi n+\pi-\epsilon}^{2\pi n+\pi+\epsilon}\sum_{\lambda:\theta_{N}(\lambda)=s,\text{ }X=\mathcal{F}(\lambda)}G(\lambda,X)ds\right].
\end{multline*}
The rest follows from a change of variable. Let us introduce $I_{\epsilon}=\frac{\pi}{2}+\cup_{n\in\mathbb{Z}}[\pi n-\epsilon,\pi n+\epsilon]$
and $\mathcal{P}_{\epsilon}(G)$:

\begin{align*}
 \mathcal{P}_{\epsilon}(G)&=\mathbb{E}\left[\frac{1}{2\epsilon}\sum_{n\in\mathbb{Z}}\int_{2\pi n-\epsilon}^{2\pi n+\epsilon}\sum_{\lambda:\theta_{N}(\lambda)=s}G(\lambda,\mathcal{F}(\lambda))ds\right]
\\
 &=\mathbb{E}\left[\int_{\mathbb{R}}G(\lambda,\mathcal{F}(\lambda))|\frac{d\theta_{N}(\lambda)}{d\lambda}|\frac{1}{2\epsilon}1_{\theta_{N}(\lambda)\in I_{\epsilon}}d\lambda\right].
\end{align*}
Then
\begin{align*} 
\frac{d\theta_{N}(\lambda)}{d\lambda} & =\frac{d\phi_{N}(\lambda)}{d\lambda}\\
 & =\frac{d}{d\lambda}\left[\prod_{k=1}^{N}T(v_{\omega}(k)-\lambda)\phi_{0}\right]\\
 & =\sum_{k=1}^{N}\frac{d}{d\phi}\left[\prod_{i=k+1}^{N}T(v_{\omega}(i)-\lambda)\right]_{v_{\omega}(N),...,v_{\omega}(k+1)}\times\\
 & \qquad\times\frac{dT(V_{\omega}(k)-\lambda)}{d\lambda}\prod_{i=1}^{k-1}T(V_{\omega}(i)-\lambda)\phi_{0}\\
 & =\sum_{k=1}^{N}\frac{d\phi_{N}}{d\phi_{k}}|_{v_{\omega}(N),...,v_{\omega}(k+1)}\cdot\frac{d}{d\lambda}[T(v_{\omega}(k)-\lambda)](\phi_{k-1}).
\end{align*}
In this formula appears the term $\frac{d\phi_{N}}{d\phi_{k}}|_{v_{\omega}(N),...,v_{\omega}(k+1)}$.
It is this term that changes the law from a forward process to a backward
process. We then calculate $\frac{d}{d\lambda}[T(v_{\omega}(k)-\lambda)](\phi_{k})\big]$
with 
\[
\begin{pmatrix}u_{k+1}\\
u_{k}
\end{pmatrix}=\begin{pmatrix}(v-\lambda)u_{k}+u_{k-1}\\
u_{k}
\end{pmatrix},
\]
\[
\frac{d}{d\lambda}\begin{pmatrix}u_{k+1}\\
u_{k}
\end{pmatrix}=\begin{pmatrix}-u_{k}\\
0
\end{pmatrix},
\]
and thus 
\[
\frac{d}{d\lambda}[T(V_{\omega}(k)-\lambda)](\phi_{k-1})\big]=\text{\ensuremath{\frac{\begin{pmatrix}u_{k+1}\\
u_{k}
\end{pmatrix}\wedge\begin{pmatrix}-u_{k}\\
0
\end{pmatrix}}{\|\begin{pmatrix}u_{k+1}\\
u_{k}
\end{pmatrix}\|^{2}}}=}\frac{u_{k}^{2}}{u_{k}^{2}+u_{k+1}^{2}}=\sin^{2}\phi_{k}.
\]
We carry on the calculation, 
\begin{align*}
\mathcal{P}_{\epsilon}(G) & =\mathbb{E}\Big[\int_{\mathbb{R}}G(\lambda,\mathcal{F}(\lambda))|\frac{d\theta_{N}(\lambda)}{d\lambda}|\frac{1}{2\epsilon}1_{\theta_{N}(\lambda)\in I_{\epsilon}}d\lambda\Big]
\\
 & =\sum_{k=1}^{N}\int_{\mathbb{R}}d\lambda\Big[\int\cdots\int d\nu(v_{1})...d\nu(v_{n})G(\lambda,\mathcal{F}(\lambda))\frac{d\phi_{N}}{d\phi_{k}}\cdot\sin^{2}(\phi_{k})\Big]\frac{1}{2\epsilon}1_{\theta_{N}(\lambda)\in I_{\epsilon}}.
\end{align*}
We artificially add a variable $\phi$ as follows:
\begin{align*}
\mathcal{P}_{\epsilon}(G) & =\sum_{k=1}^{N}\int_{\mathbb{R}}d\lambda\big[\int...\int d\nu(v_{1})...d\nu(v_{k})\int_{\mathbb{S}^{1}}d\phi\delta_{\phi_{k}}(\phi)\\
 & \qquad\int...\int d\nu(v_{k+1})...d\nu(v_{N})G(\lambda,\mathcal{F}(\lambda))\frac{d\phi_{N}}{d\phi}\cdot\sin^{2}(\phi_{k})\big]\frac{1}{2\epsilon}1_{\theta_{N}(\lambda)\in I_{\epsilon}}.
\end{align*}
Then we use Remark \ref{rem:Change} and get
\begin{align*}
\mathcal{P}_{\epsilon}(G) & =\sum_{k=0}^{N}\int_{\mathbb{R}}d\lambda\Bigg[\int...\int d\nu(v_{1})...d\nu(v_{k})\int_{\mathbb{S}^{1}}\\
 & \qquad\int...\int d\phi_{N}d\nu(v_{k+1})...d\nu(v_{N})\delta_{\phi_{k}}(\phi)G(\lambda,\mathcal{F}(\lambda))\sin^{2}(\phi_{k})\Bigg]\frac{1}{2\epsilon}1_{\theta_{N}(\lambda)\in I_{\epsilon}}\\
 & =\sum_{k=0}^{N}\int_{\mathbb{R}}d\lambda \mathbb{E}_{\mathcal{P}_{f,1..k}\otimes\mathcal{P}_{b,k+1,...,N}^{u}}\left[G(\lambda,X)\delta_{\phi_{k}^{f}-\phi_{k}^{b}}\sin^{2}(\phi_{k})\frac{1}{2\epsilon}1_{\phi_{N}\in I_{\epsilon}/\pi\mathbb{Z}}\right]
\end{align*}
where $\mathcal{P}_{f,1..k}\otimes\mathcal{P}_{b,k+1,...,N}^{u}$ is
the forward--backward process with $\mu_{b}$ the uniform law on $\mathbb{S}^{1}$. 
We can then conclude, by taking the limit $\frac{1}{2\epsilon}1_{\phi_{N}\in I_{\epsilon}/2\pi\mathbb{Z}}d\phi_{N}\rightarrow\delta_{0}$. 
\end{proof}
\bibliographystyle{siam}
%\bibliography{/home/ducatez/Documents/Travail_de_thèse/BiblioLocalisationAnderson}
\bibliography{BiblioLocalisationAnderson}

\begin{thebibliography}{10}

\bibitem{PhysRev.109.1492}
{\sc P.~W. Anderson}, {\em Absence of diffusion in certain random lattices},
  Physical Review, 109 (1958), pp.~1492--1505.

\bibitem{carmona1987anderson}
{\sc R.~Carmona, A.~Klein, and F.~Martinelli}, {\em Anderson localization for
  {B}ernoulli and other singular potentials}, Communications in Mathematical
  Physics, 108 (1987), pp.~41--66.

\bibitem{carmona2012spectral}
{\sc R.~Carmona and J.~Lacroix}, {\em Spectral theory of random Schr{\"o}dinger
  operators}, Springer Science \& Business Media, 2012.

\bibitem{cycon2009schrodinger}
{\sc H.~L. Cycon, R.~G. Froese, W.~Kirsch, and B.~Simon}, {\em Schr{\"o}dinger
  operators: With application to quantum mechanics and global geometry},
  Springer, 2009.

\bibitem{2017JSP...tmp...71D}
{\sc W.~{De Roeck}, A.~{Dhar}, F.~{Huveneers}, and M.~{Sch{\"u}tz}}, {\em {Step
  Density Profiles in Localized Chains}}, Journal of Statistical Physics,
  (2017).

\bibitem{frohlich1983}
{\sc J.~Fr{\"o}hlich and T.~Spencer}, {\em Absence of diffusion in the anderson
  tight binding model for large disorder or low energy}, Comm. Math. Phys., 88
  (1983), pp.~151--184.

\bibitem{goldsheid1977random}
{\sc I.~Goldsheid, S.~Molchanov, and L.~Pastur}, {\em A random homogeneous
  {S}chr{\"o}dinger operator has a pure point spectrum}, Functional Analysis
  and its Applications, 11 (1977), pp.~1--10.

\bibitem{kunz1980spectre}
{\sc H.~Kunz and B.~Souillard}, {\em Sur le spectre des op{\'e}rateurs aux
  diff{\'e}rences finies al{\'e}atoires}, Communications in Mathematical
  Physics, 78 (1980), pp.~201--246.

\bibitem{le1982theoremes}
{\sc {\'E}.~Le~Page}, {\em Th{\'e}oremes limites pour les produits de matrices
  al{\'e}atoires}, in Probability measures on groups, Springer, 1982,
  pp.~258--303.

\bibitem{rifkind2016eigenvectors}
{\sc B.~Rifkind and B.~Virag}, {\em Eigenvectors of the critical 1-dimensional
  random schroedinger operator}, arXiv preprint arXiv:1605.00118,  (2016).

\end{thebibliography}

\end{document}